\documentclass[11pt]{article}
\usepackage{xspace}
\usepackage{ifthen}
\usepackage{fullpage}
\usepackage[section]{definitions}
\usepackage[dvips]{color}
\usepackage{graphicx}
\usepackage{epsf}

\newcommand{\vct}[1]{\ensuremath{\mathbf{#1}}\xspace}

\newcommand{\ALG}{\ensuremath{A}\xspace}
\newcommand{\ALGP}{\ensuremath{A'}\xspace}

\newcommand{\ID}{\ensuremath{\mathrm{id}}\xspace}
\newcommand{\Id}[2]{\ensuremath{\ID(#1,#2)}\xspace}
\newcommand{\PRED}{\ensuremath{g}\xspace}
\newcommand{\Pred}[2]{\ensuremath{\PRED(#1,#2)}\xspace}

\newcommand{\DETALGS}{\ensuremath{\mathcal{D}}\xspace}
\newcommand{\ALLALGS}{\ensuremath{\mathcal{A}}\xspace}
\newcommand{\IDALGS}{\ensuremath{\mathcal{A}^I}\xspace}
\newcommand{\OPTALG}{\ensuremath{\ALG^*}\xspace}
\newcommand{\OPTPRB}{\ensuremath{\PRB^*}\xspace}
\newcommand{\OPTPRED}{\ensuremath{\PRED^*}\xspace}

\newcommand{\FE}[1][]{\ensuremath{%
\ifthenelse{\equal{#1}{}}{f}{f_{#1}}}\xspace}
\newcommand{\Fe}[2][]{\ensuremath{\FE[#1](#2)}\xspace}
\newcommand{\FEP}{\ensuremath{f'}\xspace}
\newcommand{\FeP}[1]{\ensuremath{\FEP(#1)}\xspace}
\newcommand{\FEPP}{\ensuremath{f''}\xspace}
\newcommand{\FePP}[1]{\ensuremath{\FEPP(#1)}\xspace}
\newcommand{\FEA}{\ensuremath{f}\xspace}
\newcommand{\FeA}[1]{\ensuremath{\FEA(#1)}\xspace}
\newcommand{\FEB}{\ensuremath{f'}\xspace}
\newcommand{\FeB}[1]{\ensuremath{\FEB(#1)}\xspace}

\newcommand{\Avg}[1]{\ensuremath{\overline{#1}}\xspace}
\newcommand{\FLIP}[1]{\ensuremath{\underline{#1}}\xspace}
\newcommand{\Flip}[2]{\ensuremath{\FLIP{#1}(#2)}\xspace}
\newcommand{\EQC}[1]{\ensuremath{\tilde{#1}}\xspace}

\newcommand{\LCSET}[1][]{\ensuremath{%
\ifthenelse{\equal{#1}{}}{L}{L_{#1}}}\xspace}
\newcommand{\FLIPSET}{\ensuremath{\underline{L}}\xspace}
\newcommand{\QSET}[1][]{\ensuremath{%
\ifthenelse{\equal{#1}{}}{Q}{Q_{#1}}}\xspace}

\newcommand{\DISTR}{\ensuremath{q}\xspace}
\newcommand{\Distr}[1]{\ensuremath{\DISTR(#1)}\xspace}
\newcommand{\DISTRP}{\ensuremath{q'}\xspace}
\newcommand{\DistrP}[1]{\ensuremath{\DISTRP(#1)}\xspace}
\newcommand{\DISTRPP}{\ensuremath{q''}\xspace}
\newcommand{\DistrPP}[1]{\ensuremath{\DISTRPP(#1)}\xspace}
\newcommand{\EDISTR}{\ensuremath{\tilde{q}}\xspace}
\newcommand{\EDistr}[1]{\ensuremath{\EDISTR(#1)}\xspace}

\newcommand{\ALLDISTS}{\ensuremath{\mathcal{L}}\xspace}
\newcommand{\UNIDISTS}[1][]{\ensuremath{%
\ifthenelse{\equal{#1}{}}{\mathcal{U}}{\mathcal{U}_{#1}}}\xspace}

\newcommand{\PRB}{\ensuremath{\vct{p}}\xspace}
\newcommand{\Prb}[1]{\ensuremath{p_{#1}}\xspace}

\newcommand{\PRBQ}{\ensuremath{\vct{q}}\xspace}
\newcommand{\PrbQ}[1]{\ensuremath{q_{#1}}\xspace}

\newcommand{\FunAlgErr}[2]{\ensuremath{\Delta(#1,#2)}\xspace}
\newcommand{\DPErr}[3]{\ensuremath{a[#1,#2,#3]}\xspace}
\newcommand{\Err}[3]{\ensuremath{E(#1,#2,#3)}\xspace}
\newcommand{\WErr}[2]{\ensuremath{E_w(#1,#2)}\xspace}
\newcommand{\WErrA}[1]{\ensuremath{E_w(#1)}\xspace}
\newcommand{\WErrd}[2]{\ensuremath{E_w(#1,#2)}\xspace}

\newcommand{\MaxErrQ}[1]{\ensuremath{E_2(#1)}\xspace}
\newcommand{\MaxErr}[1]{\ensuremath{E_1(#1)}\xspace}

\newcommand{\LPERR}{\ensuremath{Z}\xspace}
\newcommand{\LPERRQ}{\ensuremath{\hat{Z}}\xspace}

\newcommand{\METSPACE}{\ensuremath{\mathcal{M}}\xspace}
\newcommand{\DDIM}{\ensuremath{\beta}\xspace}
\newcommand{\DIST}{\ensuremath{d}\xspace}
\newcommand{\Dist}[2]{\ensuremath{\DIST(#1,#2)}\xspace}

\newcommand{\MAXRANGE}{\ensuremath{r}\xspace}

\newcommand{\MEDV}{\ensuremath{m}\xspace}
\newcommand{\MEDP}{\ensuremath{o}\xspace}

\newcommand{\Ring}[1]{\ensuremath{R_{#1}}\xspace}
\newcommand{\PRing}[1]{\ensuremath{\iota(#1)}\xspace}
\newcommand{\RDist}[1]{\ensuremath{2^{-#1}}\xspace}
\newcommand{\GbR}[2]{\ensuremath{o_{#1,#2}}\xspace}
\newcommand{\Gb}[2]{\ensuremath{B_{#1,#2}}\xspace}

\newcommand{\QUANT}{\ensuremath{\gamma}\xspace}

\newcommand{\PEAK}{\ensuremath{b}\xspace}
\newcommand{\ZS}[2]{\ensuremath{Z^{#1}_{#2}}\xspace}

\newcommand{\MinVal}[2]{\ensuremath{{L_{min}}(#1,#2)}\xspace}
\newcommand{\MaxVal}[2]{\ensuremath{{L_{max}}(#1,#2)}\xspace}

\newcommand{\OPT}{\ensuremath{\mathrm{OPT}}\xspace}

\begin{document}

\title{Estimating the Average of a
Lipschitz-Continuous Function from One Sample
}
\vspace{0.8cm}
\author{Abhimanyu Das\\
University of Southern California\\
{\tt abhimand@usc.edu}
\and
David Kempe\thanks{Supported in part by NSF CAREER award 0545855, and
  NSF grant DDDAS-TMRP 0540420}\\
University of Southern California\\
{\tt \dkempeemail}
}


\maketitle

\begin{abstract}
We study the problem of estimating the average of a Lipschitz continuous
function $f$ defined over a metric space, by querying $f$ at only a single point.
More specifically, we explore the role of
randomness in drawing this sample. Our goal is to find a distribution
minimizing the expected estimation error against an adversarially
chosen Lipschitz continuous function.
Our work falls into the broad class of estimating
aggregate statistics of a function from a small number of carefully
chosen samples. The general problem has a wide
range of practical applications in areas as diverse as
sensor networks, social sciences and numerical analysis.
However, traditional work in numerical analysis has focused
on asymptotic bounds, whereas we are interested in the \emph{best}
algorithm.
For arbitrary discrete metric spaces of bounded doubling dimension, we
obtain a PTAS for this problem. In the special case when the points
lie on a line, the running time improves to an FPTAS.
Both algorithms are based on approximately solving a linear program
with an infinite set of constraints, by using an approximate
separation oracle.
For Lipschitz-continuous functions over $[0,1]$, we
calculate the precise achievable error as
$1-\frac{\sqrt{3}}{2} \approx 0.134$, which improves upon the \quarter
which is best possible for deterministic algorithms.
\end{abstract}

\section{Introduction}
\label{sec:introduction}

One of the most fundamental problems in data-driven sciences is to
estimate some aggregate statistic of a real-valued function 
\FE by sampling \FE in few places. 
Frequently, obtaining samples incurs a cost in terms of
human labor, computation, energy or time. Thus, researchers face
an inherent tradeoff between the accuracy of estimating the aggregate statistic
 and the number of samples required.
With samples a scarce resource, it becomes an important computational
problem to determine where to sample \FE, and how to post-process the
samples.

Naturally, there are many mathematical formulations of this
estimation problem, depending on the aggregate statistic that we 
wish to estimate (such as the average, median or maximum value), the
error objective that we wish to minimize (such as worst-case absolute
error, average-case squared error, etc.), 
and on the conditions imposed on the function. 
In this paper, we study algorithms optimizing a worst-case error objective, 
i.e., we assume that \FE is chosen adversarially.
Motivated by the applications described below, we use 
Lipschitz-continuity to impose a ``smoothness'' condition on \FE.
(Note that without any smoothness conditions on \FE, we cannot hope to 
approximate any aggregate function in an adversarial setting
without learning all function values.) 
That is, we assume that the domain of \FE is a metric space, and that \FE is
Lipschitz-continuous over its domain. Thus, nearby points are
guaranteed to have similar function values.

Here, we focus on perhaps the simplest aggregation function:
the average \Avg{\FE}. Despite its simplicity, it has many natural applications,
such as
\begin{enumerate}
\item In sensor networks covering a geographical area,
the average of a natural phenomenon (such as temperature or pressure)
is frequently one of the most interesting quantities. Here, nearby
locations tend to yield similar measurements. Since energy is a scarce
resource, it is desirable to sample only a few of the
deployed sensors.
\item In population surveys, researchers are frequently interested in
the average of quantities such as income or education level.
A metric on the population may be based on job similarity, which would
have strong predictive value for these quantities. Interviewing a
subject is time-consuming, and thus sample sizes tend to be much
smaller than the entire population.
\item In numerical analysis, one of the most fundamental problems is
numerical integration of a function. If the domain is continuous, this
corresponds precisely to computing the average. If the function to be
integrated is costly to evaluate, then again, it is desirable to 
sample a small number of points.
\end{enumerate}

If \FE is to be evaluated at $k$ points, 
chosen deterministically and non-adaptively, then previous
work \cite{das} shows that the optimum sampling locations
for estimating the average of \FE form a $k$-median of the metric space.
However, the problem becomes significantly more complex when the
algorithm gets to randomize its choices of sampling locations. In
fact, even the seemingly trivial case of $k=1$ turns out to be highly
non-trivial, and is the focus of this paper. Addressing this case 
 is an important step toward the ultimate goal of
  understanding the tradeoffs between the number of samples and the
  estimation error.

Formally, we thus study the following question:
Given a metric space \METSPACE, a randomized \todef{sampling algorithm} is
described by
(1) a method for sampling a location $x \in \METSPACE$ from a distribution \PRB;
(2) a function \PRED for predicting the average \Avg{\FE} of the function \FE 
over \METSPACE, using the sample $(x, \Fe{x})$.
The expected estimation error 
is then
$\Err{\PRB}{\PRED}{f} = 
\sum_{x \in \METSPACE} \Prb{x} \cdot \Abs{\Pred{x}{\Fe{x}} - \Avg{\FE}}$.
(The sum is replaced by an integral, and \PRB by a density,
if \METSPACE is continuous.)
The worst-case error is
$\WErr{\PRB}{\PRED} = 
\sup_{\FE \in \LCSET} \Err{\PRB}{\PRED}{\FE}$, 
where \LCSET is the set of all $1$-Lipschitz continuous functions defined on $\METSPACE$. 
Our goal is to find a randomized sampling algorithm 
(i.e., a distribution \PRB and function \PRED, 
computable in polynomial time) that (approximately) minimizes
\WErr{\PRB}{\PRED}.

In this paper, we provide a PTAS for this problem of 
minimizing \WErr{\PRB}{\PRED}, for any discrete metric
space \METSPACE with constant doubling dimension. (This 
includes constant-dimensional Euclidean metric spaces.) 
For discrete metric
spaces \METSPACE embedded on a line, we improve this result to an
FPTAS. Both of these algorithms are based on a linear program with
infinitely many constraints, for which an approximate separation
oracle is obtained. 

We next study the perhaps simplest variant of this problem, in
which the metric space is the interval $[0,1]$. While the worst-case
error of any deterministic algorithm is obviously \quarter in this
case, we show that for a randomized algorithm, the bound improves to
$1-\frac{\sqrt{3}}{2}$. We prove this by providing
an explicit distribution, and obtaining a matching lower bound using
Yao's Minimax Principle. Our result can also be interpreted as showing
how ``close'' a collection of Lipschitz-continuous functions on $[0,1]$
must be.

\subsection{Related Work}
Estimating the integral of a smooth function \FE using its values 
at a discrete set of points is one of the core problems in numerical analysis.
The tradeoffs between the number of samples needed and the
estimation error bounds have been investigated in detail under the name of
\todef{Information Based Complexity (IBC)}~\cite{traub2,traub}.
More generally, IBC studies the problem of computing approximations
to an operator $S(f)$ on functions $f$ from a set $F$
(with certain ``smoothness'' properties) using a finite set of samples
$N(f)=[L_1(f), L_2(f), \ldots, L_n(f)]$. The $L_i$ are functionals.
For a given algorithm $U$, its error is
$E(U)=\sup_{f\in F} \Norm{S(f) - U(f)}$.
The goal in IBC is to find an $\epsilon$-approximation $U$ (i.e.,
ensuring that $E(U) \leq \epsilon$) with least information cost $c(U)=n$.

One of the common problems in IBC is multivariate integration of
real-valued functions with a smoothness parameter $r$ over $d$-dimensional
unit balls.
For such problems, Bakhvalov~\cite{bakhvalov} designed a randomized algorithm
providing an $\epsilon$-approximation with cost
$\Theta(\frac{1}{\epsilon^{2d/(d+2r)}})$.
Bakhvalov~\cite{bakhvalov} and Novak~\cite{novak} also show that this 
cost is asymptotically optimal.
The papers by  Novak~\cite{novak} and Mathe~\cite{mathe} show that if $r=0$,
then simple Monte-Carlo integration algorithms
(which sample from the uniform distribution) have an asymptotically
optimal cost of $\frac{1}{\epsilon^2}$.

In \cite{wozniakowski,wozniakowski2}, Wozniakowski studied the average case
complexity of linear multivariate IBC problems,
and derived conditions under which the problems are tractable,
i.e., have cost polynomial in $\frac{1}{\epsilon}$ and $d$.
Wojtaszczyk~\cite{Wojtaszczyk} proved that the multivariate
integration problem is
not strongly tractable (polynomial in $\frac{1}{\epsilon}$ and independent of $d$).


In \cite{baran}, Baran et al.~study the IBC problem in the univariate
integration model for Lipschitz continuous functions, and formulate
approximation bounds in an adaptive setting.
That is, the sampling strategy can change adaptively based on the
previously sampled values.
They provide deterministic and randomized $\epsilon$-approximation algorithms,
which, for any problem instance $P$, use
$O(\log(\frac{1}{\epsilon \cdot \OPT}) \cdot \OPT)$ samples for the deterministic
case and $O(\OPT^{4/3} + \OPT \cdot \log(\frac{1}{\epsilon}))$ samples for the 
randomized case. Here, $\OPT$ is the optimal number of
samples for the problem instance $P$. They prove that their algorithms are
asymptotically optimal, compared to any other adaptive algorithm.

There are two main differences between the results in IBC and our
work: first, IBC treats the target approximation as given and the
number of samples as the quantity to be minimized. Our goal is to
minimize the expected worst-case error with a fixed number of samples
(one). More importantly, results in IBC are traditionally
\emph{asymptotic}, ignoring constants. For a single sample, this would
trivialize the problem: it is implicit in our proofs that sampling at
the metric space's median is a constant-factor approximation to the
best randomized algorithm.

The deterministic version of our problem was studied previously in
\cite{das}. There, it was shown that the best sampling locations for
reading $k$ values non-adaptively constitute the optimal $k$-median of the
metric space. Thus, the algorithm of Arya et al.~\cite{arya}
gives a polynomial-time $(3 + \epsilon)$-approximation algorithm to
identify the best $k$ values to read.

\section{Preliminaries} \label{sec:preliminaries}
We are interested in real-valued
Lipschitz-continuous functions over metric spaces
of constant doubling dimension (e.g., \cite{gupta}). 
Let $(\METSPACE,\DIST)$ be a compact metric space with distances
\Dist{x}{y} between pairs of points. W.l.o.g., we assume that
$\max_{x,y \in \METSPACE} \Dist{x}{y} = 1$.
We require $(\METSPACE,\DIST)$ to have constant doubling dimension
\DDIM, i.e., for every $\delta$, each ball of diameter $\delta$ can be
covered by at most $c^\DDIM$ balls of diameter $\delta/c$, for any $c \geq 2$.

A real-valued function \FE is Lipschitz-continuous (with constant 1)
if $\Abs{\Fe{x} - \Fe{y}} \leq \Dist{x}{y}$ for all points $x, y$. We define 
\LCSET to be the set of all such Lipschitz-continuous functions \FE, i.e., 
$\LCSET = \Set{\FE}{\Abs{\Fe{x} - \Fe{y}} \leq \Dist{x}{y} 
\mbox{ for all } x,y}$. Since we will frequently want to bound the
function values, we also define
$\LCSET[c] = \Set{\FE \in \LCSET}{\Abs{\int_x \Fe{x} dx} \leq c}$.
Notice that \LCSET[c] is a compact set.

We wish to predict the average $\Avg{\FE} = \int_x \Fe{x} dx$ of
all the function values. 
When \METSPACE is finite of size $n$, then the average is of course
$\Avg{\FE} = \frac{1}{n} \cdot \sum_x \Fe{x}$ instead.
The algorithm first gets to choose a single point $x$ according to a
(polynomial-time computable) density function \PRB;
it then learns the value \Fe{x}, and may post-process it with 
a \todef{prediction function} \Pred{x}{\Fe{x}} to produce its estimate
of the average \Avg{\FE}.
The goal is to minimize the expected estimation error of the average,
assuming \FE is chosen adversarially from \LCSET with knowledge of the
algorithm, but not its random choices.
Formally, the goal is to minimize
$\WErrd{\PRB}{\PRED} = \sup_{\FE \in \LCSET} 
(\int_x \Prb{x} \cdot \Abs{\Avg{\FE} - \Pred{x}{\Fe{x}}} dx)$.
If \METSPACE is finite, then \PRB will be a probability distribution
instead of a density, and the error can now be written as  
$\WErrd{\PRB}{\PRED}= \sup_{\FE \in \LCSET} 
(\sum_x \Prb{x} \cdot \Abs{\Avg{\FE} - \Pred{x}{\Fe{x}}})$.

Formally, we consider an algorithm to be the pair $(\PRB, \PRED)$ of
the distribution and prediction function. Let \ALLALGS denote the set
of all such pairs, and \DETALGS the set of all \emph{deterministic}
algorithms, i.e., algorithms for which \PRB has all its density on a
single point.
Our analysis will make heavy use of Yao's Minimax
Principle \cite{motwani:raghavan}. To state it, we
define \ALLDISTS to be the set of all probability distributions over
\LCSET.  We also define the estimation error
$\FunAlgErr{\FE}{\ALG}=\int_x \Prb{x} \cdot \Abs{\Avg{\FE} - \Pred{x}{\Fe{x}}} dx$,
where \ALG corresponds to the pair $(\PRB, \PRED)$. 

\begin{theorem}[Yao's Minimax Principle \cite{motwani:raghavan}]
\[ \begin{array}{lcl}
\sup_{\DISTR \in \ALLDISTS} \inf_{\ALG \in \DETALGS}
\Expect[\FE \sim \DISTR]{\FunAlgErr{\FE}{\ALG}}
 &=& \inf_{\ALG \in \ALLALGS} \sup_{\FE \in \LCSET}
\FunAlgErr{\FE}{\ALG}.
\end{array} \]

\end{theorem}

We first show that without loss of
generality, we can focus on algorithms 
whose post-processing is just to output the observed value, i.e.,
algorithms $(\PRB, \ID)$ with $\Id{x}{y} = y$, for all $x,y$.
When \PRED is the identity function, we simply write 
$\FunAlgErr{\FE}{\PRB}=\int_x \Prb{x} \cdot \Abs{\Avg{\FE} - \Fe{x}} dx$
for the error incurred by using the distribution \PRB.

\begin{theorem} \label{thm:randpredictionfunction}
Let $\OPTALG = (\OPTPRB, \OPTPRED)$ be the optimum randomized
algorithm. Then, for every $\epsilon > 0$, there is a randomized
algorithm $\ALG = (\PRB, \ID)$, such that 
$\WErrA{\ALG} \leq \WErrA{\OPTALG} + \epsilon$.
\end{theorem}

\begin{proof}
Let \IDALGS denote the set of all (randomized) algorithms using
the identity function for post-processing, i.e.,
$\IDALGS = \Set{\ALG=(\PRB,\ID)}{\PRB \mbox{ is a distribution over \METSPACE}}$.

For the analysis, we are interested in equivalence classes of
functions; we say that $\FE, \FEP$ are \todef{equivalent} if either
(1) there exists a constant $c$ such that
$\Fe{x} = c + \FeP{x}$ for all $x$, or 
(2) there is a constant $c$ such that 
$\Fe{x} = c- \FeP{x}$ for all $x$. 
Let \EQC{\FE} denote the equivalence class of \FE, i.e., the set of
all vertical translations of \FE and its vertical reflection.

Given $\epsilon$, we let \MAXRANGE be a large enough constant
defined below.
A distribution \DISTR
over \LCSET[\MAXRANGE] is called \todef{equivalence-uniform} if the
distribution, restricted to any equivalence class, is
uniform\footnote{Unfortunately, this definition does not extend to
  \LCSET, since \EQC{\FE} is not bounded, and a uniform distribution
  is thus not defined. This issue causes the $\epsilon$ terms in
  the theorem.}.
That is, for any $\FEP \in \EQC{\FE}$ with $\FEP, \FE \in \LCSET[\MAXRANGE]$, 
we have $\Distr{\FEP} = \Distr{\FE}$.
Let \UNIDISTS[\MAXRANGE] denote the set of all equivalence-uniform distributions
over \LCSET[\MAXRANGE].
We will show two facts:
\begin{enumerate}
\item If $\DISTR \in \UNIDISTS[\MAXRANGE]$, then for any deterministic
  algorithm $\ALG \in \DETALGS$, there is a deterministic algorithm
  $\ALGP \in \DETALGS \cap \IDALGS$ which outputs simply the value it
  sees, such that
\LEquation[eqn:identity-alg]{%
\Expect[\FE \sim \DISTR]{\FunAlgErr{\FE}{\ALGP}}}{%
\Expect[\FE \sim \DISTR]{\FunAlgErr{\FE}{\ALG}} + \epsilon/2.}
\item For any distribution $\DISTR \in \ALLDISTS$ of
Lipschitz-continuous functions, there is an equivalence-uniform
distribution $\DISTRP \in \UNIDISTS[\MAXRANGE]$ (where \MAXRANGE may
depend on \DISTR) such that
\LEquation[eqn:symmetric-dist]{%
\inf_{\ALG \in \DETALGS \cap \IDALGS} \Expect[\FE \sim \DISTR]{\FunAlgErr{\FE}{\ALG}}}{%
\inf_{\ALG \in \DETALGS \cap \IDALGS} \Expect[\FE \sim \DISTRP]{\FunAlgErr{\FE}{\ALG}} + \epsilon/2.}
\end{enumerate}

Using these two inequalities, and applying Yao's Minimax Theorem
twice then completes the proof as follows:

\begin{eqnarray*}
\inf_{\ALG \in \IDALGS} \sup_{\FE \in \LCSET} \FunAlgErr{\FE}{\ALG}
& = & 
\sup_{\DISTR \in \ALLDISTS} \inf_{\ALG \in \DETALGS \cap \IDALGS} 
\Expect[\FE \sim \DISTR]{\FunAlgErr{\FE}{\ALG}}\\
& \stackrel{(\ref{eqn:symmetric-dist})}{\leq} & 
\sup_{\DISTR \in \UNIDISTS[\MAXRANGE]} \inf_{\ALG \in \DETALGS \cap \IDALGS} 
\Expect[\FE \sim \DISTR]{\FunAlgErr{\FE}{\ALG}} + \epsilon/2\\
& \stackrel{(\ref{eqn:identity-alg})}{\leq} & 
\sup_{\DISTR \in \UNIDISTS[\MAXRANGE]} \inf_{\ALG \in \DETALGS} 
\Expect[\FE \sim \DISTR]{\FunAlgErr{\FE}{\ALG}} + \epsilon\\
& \leq & 
\sup_{\DISTR \in \ALLDISTS} \inf_{\ALG \in \DETALGS} 
\Expect[\FE \sim \DISTR]{\FunAlgErr{\FE}{\ALG}} + \epsilon\\
& = & 
\inf_{\ALG \in \ALLALGS} \sup_{\FE \in \LCSET} \FunAlgErr{\FE}{\ALG} + \epsilon.
\end{eqnarray*}

It thus remains to show the two inequalities. We begin with Inequality
(\ref{eqn:identity-alg}). Let $x$ be the point at which \ALG samples
the function. For any function $\FE \in \LCSET[\MAXRANGE]$
let \FLIP{\FE} be the ``flipped'' function around $x$, defined by
$\Flip{\FE}{y} = 2\Fe{x} - \Fe{y}$ for all $y$.
Let \FLIPSET be the set of all functions \FE such that both \FE and
\FLIP{\FE} are in \LCSET[\MAXRANGE]. 
Because $\Avg{\FLIP{\FE}} = 2\Fe{x} - \Avg{\FE}$
 and $\Abs{\Fe{x}-\Avg{\FE}} \leq \half$ for all $x$,
we obtain that $\FLIPSET \supseteq \LCSET[\MAXRANGE-1]$.
Also, because $\FLIP{\FE} \in \EQC{\FE}$ and
$\DISTR \in \UNIDISTS[\MAXRANGE]$, we have $\Distr{\FE} = \Distr{\FLIP{\FE}}$
whenever $\FE \in \FLIPSET$.
We thus obtain that 
\begin{eqnarray*}
\Expect[\FE \sim \DISTR]{\FunAlgErr{\FE}{\ALG}}
& = & \int_{\FLIPSET} \Abs{\Pred{x}{\Fe{x}} - \Avg{\FE}} \Distr{\FE} d\FE
+ \int_{\Compl{\FLIPSET}} \Abs{\Pred{x}{\Fe{x}} - \Avg{\FE}} \Distr{\FE} d\FE\\
& \geq & 
\half \int_{\FLIPSET} (\Abs{\Pred{x}{\Fe{x}} - \Avg{\FE}} +
                      \Abs{\Pred{x}{\Flip{\FE}{x}} - \Avg{\FLIP{\FE}}})
\Distr{\FE} d\FE\\
& = & 
\half \int_{\FLIPSET} (\Abs{\Pred{x}{\Fe{x}} - \Avg{\FE}} +
                       \Abs{\Pred{x}{\Fe{x}} - \Avg{\FLIP{\FE}}})
\Distr{\FE} d\FE\\
& \geq & 
\half \int_{\FLIPSET} \Abs{\Avg{\FE} - \Avg{\FLIP{\FE}}} \Distr{\FE} d\FE.
\end{eqnarray*}
For the first inequality, we dropped the second integral, and used the
symmetry of the distribution to write the first integral twice and
then regroup. The second step used that by definition, 
$\Flip{\FE}{x} = \Fe{x}$, and the third the inverse triangle inequality.
By definition, \Fe{x} lies between \Avg{\FE} and \Avg{\FLIP{\FE}};
therefore, $\Abs{\Avg{\FE} - \Avg{\FLIP{\FE}}} = 
\Abs{\Avg{\FE} - \Fe{x}} + \Abs{\Fe{x} - \Avg{\FLIP{\FE}}}$, and we can
  further bound
\begin{eqnarray*}
\Expect[\FE \sim \DISTR]{\FunAlgErr{\FE}{\ALG}}
& \geq & 
\half \int_{\FLIPSET} \Abs{\Avg{\FE} - \Fe{x}} + \Abs{\Fe{x} - \Avg{\FLIP{\FE}}} \Distr{\FE} d\FE\\
& = & 
\int_{\FLIPSET} \Abs{\Avg{\FE} - \Fe{x}} \Distr{\FE} d\FE\\
& \geq & 
\int_{\LCSET[\MAXRANGE]} \Abs{\Avg{\FE} - \Fe{x}} \Distr{\FE} d\FE - \epsilon/2,
\end{eqnarray*}
because symmetry of \DISTR implies that 
$\Prob{\FE \notin \FLIPSET} \leq 1/\MAXRANGE \leq \epsilon/2$ (we will
set $\MAXRANGE \geq 2/\epsilon$), and Lipschitz
continuity implies that $\Abs{\Avg{\FE} - \Fe{x}} \leq 1$ for all $x$.

\smallskip

Next, we prove Inequality (\ref{eqn:symmetric-dist}). Let \DISTR be an
arbitrary distribution, and \MAXRANGE large enough such that 
$\Prob[\FE \sim \DISTR]{\FE \notin \LCSET[\MAXRANGE]} \leq \epsilon/2$.
First, we truncate \DISTR to a distribution \DISTRPP over \LCSET[\MAXRANGE]:
We set $\DistrPP{\FE} = 0$ for all $\FE \notin \LCSET[\MAXRANGE]$, and
renormalize by setting 
$\DistrPP{\FE} =  \frac{1}{\Prob[\FE \sim \DISTR]{\FE \in \LCSET[\MAXRANGE]}}
\cdot \Distr{\FE}$ for $\FE \in \LCSET[\MAXRANGE]$.
Next, we define a distribution \EDISTR over equivalence classes
\EQC{\FE} as 
$\EDistr{\EQC{\FE}} = \int_{\FEP \in \EQC{\FE}} \DistrPP{\FEP} d\FEP$;
finally, let \DISTRP be defined by choosing an equivalence class
\EQC{\FE} according to \EDISTR, and subsequently choosing a member of
$\EQC{\FE} \cap \LCSET[\MAXRANGE]$ uniformly at random; clearly,
\DISTRP is equivalence-uniform.
Let $\ALG \in \argmin_{\ALG \in \DETALGS \cap \IDALGS} 
\Expect[\FE \sim \DISTRP]{\FunAlgErr{\FE}{\ALG}}$ be a deterministic
algorithm with identity post-processing function 
minimizing the expected estimation error for \DISTRP; 
let $x$ be the point at which \ALG samples the function.

Since the algorithm always samples at $x$ and outputs \FeP{x},
the estimation error \Abs{\FeP{x} - \Avg{\FEP}} is the same
for all $\FEP \in \EQC{\FE}$, because all these \FEP are simply
shifted or mirrored from each other.
If used instead on the initial distribution \DISTR, \ALG has expected
estimation error 
\begin{eqnarray*}
\int_{\LCSET} \Abs{\Fe{x} - \Avg{\FE}} \Distr{\FE} d\FE
& \leq & \Prob[\FE \sim \DISTR]{\FE \notin \LCSET[\MAXRANGE]} 
+ \int_{\LCSET[\MAXRANGE]} \Abs{\Fe{x} - \Avg{\FE}} \DistrPP{\FE} d\FE\\
& \leq & \epsilon/2 + 
\int \int_{\FEP \in \EQC{\FE}} \Abs{\FeP{x} - \Avg{\FEP}} 
\Distr{\FEP} d\FEP d\EQC{\FE}\\
& = & \epsilon/2 + 
\int \Abs{\Fe{x} - \Avg{\FE}} \EDistr{\EQC{\FE}} d\EQC{\FE}\\
& = & \epsilon/2 +
\int \int_{\FEP \in \EQC{\FE}} \Abs{\FeP{x} - \Avg{\FEP}} 
\DistrP{\FEP} d\FEP d\EQC{\FE}\\
& = & \epsilon/2 + 
\int_{\LCSET[\MAXRANGE]} \Abs{\Fe{x} - \Avg{\FE}} \DistrP{\FE} d\FE.
\end{eqnarray*}
The inequality in the first step came from upper-bounding the
estimation error outside \LCSET[\MAXRANGE] by 1, and using that
$\DistrPP{\FE} \geq \Distr{\FE}$.
\end{proof}


\Omit{
For a given algorithm with probability distribution \Prb{i} and
prediction function \PRED, we call \SE a \todef{worst-case vector} if
it maximizes $\sum_i \Prb{i} \Abs{\Avg{\SE} - \Pred{i}{\Se{i}}}$.
We will sometimes assume that all vectors \SE are normalized such
that $\Avg{\SE}=0$, in which case 
$\WErrd{\PRB}{\PRED}  = 
\max_{\SE \in \LCSET and \Avg{\SE}=0 }(\sum_{i=1}^n \Prb{x}\Abs{\Pred{i}{\SE{x}}}dx)$. 

We define the index $m$ to be the $1$-median location of all the $n$ points, and
$d_m$ to be the $1$-median distance $\frac{\sum_{i=1}^{n} \Dist{i}{m}}{n}$

For any point $S_j$, we define its maximizing vector $\MV{j}$ as 
$\Mv{j}{i} = \Dist{i}{j}$.

The following definitions below are used for the related problem of
finding the best deterministic algorithm (that samples at a location s, and outputs
an estimate \Pred{s}{\Se{i}} for minimizing the expected error)
in prediction of the average over an input distribution of \LVALID vectors (denoted
by the probability distribution \PRB).

We define \PE{\PRB}{\Alg} to be the expected estimation error for the input 
probability
distribution $\PRB$ over all \LVALID vectors, using a deterministic algorithm \DAlg. 
Also $\PEm{\PRB} = \max_{\DAlg} (\PE{\PRB}{\DAlg})$ is the expected estimation error for 
the distribution \PRB using the best possible deterministic algorithm.

Hence, the aggregate function estimation 
problem can now be rephrased as follows:
\begin{definition}[Discrete aggregate function estimation]
Given a set of $n$ points embedded in a metric space with doubling constant \DDIM
and distance function $\DIST$, 
find the probability distribution \PRB,
which minimizes the
expected worst-case estimation error $\WErrd{\PRB}{\PREDI}=
\max_{\SE \in \LCSET} \Errd{\PRB}{\PREDI}{\SE}$, where
$\Errd{\PRB}{\PREDI}{\SE}=\sum_{i=1}^n \Prb{i} \Abs{\Avg{\SE} - \Se{i}}$ is 
the estimation error for a vector \SE.
\end{definition}
}

\section{Discrete Metric Spaces} \label{sec:discrete}
In this section, we focus on finite metric spaces, consisting of $n$
points. Thus, instead of integrals and densities, we will be
considering sums and probability distributions. The characterization
of using the identity function for post-processing from 
Theorem \ref{thm:randpredictionfunction} holds in this case as well;
hence, without loss of generality, we assume that all
algorithms simply output the value they observe.
The problem of finding the best probability distribution for a single
sample can be expressed as a linear program, with variables \Prb{x}
for the sampling probabilities at each of the $n$ points $x$, and a
variable \LPERR for the estimation error.

\begin{LP}[eqn:exact-lp]{Minimize}{\LPERR}
(\mathrm{i})\;\;\; \sum_x \Prb{x} = 1\\

(\mathrm{ii})\;\; \sum_x \Prb{x} \cdot \Abs{\Avg{\FE} - \Fe{x}} \leq \LPERR
& \mbox{ for all } \FE \in \LCSET\\

(\mathrm{iii})\; 0 \leq \Prb{x} \leq  1 
& \mbox{ for all points } x
\end{LP}

Since this LP (which we refer to as the ``exact LP'') has infinitely many constraints, our approach is to
replace the set \LCSET in the second constraint with a set 
\QSET[\delta]. We will choose \QSET[\delta] carefully to ensure
that it ``approximates'' \LCSET well, and such that the resulting LP
below (which we refer to as the ``discretized LP'') 
can be solved efficiently.

\begin{LP}[eqn:approx-lp]{Minimize}{\LPERRQ}
(\mathrm{i})\;\;\; \sum_x \Prb{x} = 1\\

(\mathrm{ii})\;\; \sum_x \Prb{x} \cdot \Abs{\Avg{\FE} - \Fe{x}} \leq \LPERRQ
& \mbox{ for all } \FE \in \QSET[\delta]\\

(\mathrm{iii})\; 0 \leq \Prb{x} \leq  1 
& \mbox{ for all points } x
\end{LP}

To define the notion of approximation formally, 
let \MEDP be a 1-median of the metric space, 
i.e., a point minimizing $\sum_x \Dist{\MEDP}{x}$.
Let $\MEDV = \frac{1}{n} \sum_x \Dist{\MEDP}{x}$ be the average
distance of all points from \MEDP. 

Because we assumed w.l.o.g.~that
$\max_{x,y \in \METSPACE} \Dist{x}{y} = 1$, 
at least one point has distance at least \half from \MEDP, and
therefore, $\MEDV \geq \frac{1}{2n}$.
The median value \MEDV forms a lower bound for randomized algorithms in the following sense.

\begin{lemma} \label{lem:lowerboundkd}
The worst-case expected error for any randomized algorithm is at 
least $\frac{1}{4 \cdot 6^\DDIM} \cdot \MEDV$, where \DDIM is the doubling dimension of the metric space.
\end{lemma}

\begin{emptyproof}
Consider any randomized algorithm with probability distribution \PRB;
w.l.o.g., the algorithm outputs the value it observes.
Let $R = \Set{x}{\frac{\MEDV}{2} \leq \Dist{x}{\MEDP} \leq \frac{3\MEDV}{2}}$
be the ring of points at distance between
$\frac{\MEDV}{2}$ and $\frac{3\MEDV}{2}$ from \MEDP.
We distinguish two cases:

\begin{enumerate}
\item If $\sum_{x \in R} \Prb{x} \leq \half$, then consider the
  Lipschitz-continuous function defined by $\Fe{x} = \Dist{x}{\MEDP}$.
  This function has average $\Avg{\FE} = \MEDV$. 
  With probability at least \half, 
  the algorithm samples a point outside $R$, and thus outputs a value 
  outside the interval $[\frac{\MEDV}{2}, \frac{3\MEDV}{2}]$, which
  incurs error at least $\frac{\MEDV}{2}$.
  Thus, the expected error is at least $\frac{\MEDV}{4}$.
\item If $\sum_{x \in R} \Prb{x} > \half$, then consider 
a collection of balls $B_1, \ldots, B_k$ of diameter $\frac{\MEDV}{2}$
covering all points in $R$. Because $R$ is contained in a ball of
diameter $3 \MEDV$, the doubling constraint implies that
$k \leq 6^\DDIM$ balls are sufficient. At least one of these
balls --- say, $B_1$ --- has $\sum_{x \in B} \Prb{x} \geq \frac{1}{2k}$.
Fix an arbitrary point $y \in B_1$, and define the Lipschitz-continuous
function \FE as $\Fe{x} = \Dist{x}{y}$.
Because \MEDP was a 1-median, we get that $\Avg{\FE} \geq \MEDV$.
With probability at least $\frac{1}{2k}$, the algorithm will choose a
point inside $B_1$ and output a value of at most $\frac{\MEDV}{2}$,
thus incurring an error of at least $\frac{\MEDV}{2}$. 
Hence, the expected error is at least
$\frac{1}{2k} \cdot \frac{\MEDV}{2} \geq \frac{1}{4 \cdot 6^\DDIM}
\cdot \MEDV$.\QED
\end{enumerate}
\end{emptyproof}

We now formalize our notion for a set of functions \QSET[\delta]
to be a good approximation. 

\begin{definition}[$\delta$-approximating function classes]
\label{def:approximating}
For any sampling distribution \PRB, define
$\MaxErr{\PRB} = \max_{\FE \in \LCSET} \FunAlgErr{\FE}{\PRB}$
and 
$\MaxErrQ{\PRB} = \max_{\FE \in \QSET[\delta]} \FunAlgErr{\FE}{\PRB}$
to be the maximum error of sampling according to \PRB against a
worst-case function from \LCSET and \QSET[\delta], respectively.
The class \QSET[\delta] is said to \todef{$\delta$-approximate} \LCSET
if the following two conditions hold:
\begin{enumerate}
\item For each $\FE \in \LCSET$, there is a function
$\FEP \in \QSET[\delta]$ such that
$\Abs{\FunAlgErr{\FEP}{\PRB} - \FunAlgErr{\FE}{\PRB}}
\leq \frac{\delta}{2} \cdot \MaxErr{\PRB}$, 
for all distributions \PRB.
\item For each $\FE \in \QSET[\delta]$, there is a function
$\FEP \in \LCSET$ such that
$\Abs{\FunAlgErr{\FEP}{\PRB} - \FunAlgErr{\FE}{\PRB}}
\leq \frac{\delta}{2} \cdot \MaxErr{\PRB}$, 
for all distributions \PRB.
\end{enumerate}
\end{definition}

\begin{theorem} \label{thm:approximate-LP}
Assume that for every $\delta$, \QSET[\delta] is a class of functions
$\delta$-approximating \LCSET, such that
the following problem can be solved in polynomial time (for fixed $\delta$):
Given \PRB, find a function $\FE \in \QSET[\delta]$ maximizing
$\FunAlgErr{\FE}{\PRB}$.

Then, solving the discretized LP (\ref{eqn:approx-lp}) instead of the 
exact LP
(\ref{eqn:exact-lp}) gives a PTAS for the problem of finding a 
sampling distribution that minimizes the worst-case expected error.
\end{theorem}

\begin{proof}
First, an algorithm to find a function $\FE \in \QSET[\delta]$ 
maximizing $\sum_x \Prb{x} \cdot \Abs{\Avg{\FE} - \Fe{x}}$
gives a separation oracle for the discretized LP.
Thus, using the Ellipsoid Method
(e.g., \cite{grotschel}), an optimal solution to
the discretized LP can be found in polynomial time, for any fixed $\delta$.

Let \PRB, \PRBQ be optimal solutions to the exact and discretized LPs,
respectively.
Let $f_1 \in \LCSET$ maximize
$\sum_x \PrbQ{x} \cdot \Abs{\Avg{\FE} - \Fe{x}}$ over $\FE \in \LCSET$,
and $f_2 \in \QSET[\delta]$ maximize
$\sum_x \Prb{x} \cdot \Abs{\Avg{\FE} - \Fe{x}}$ over $\FE \in \QSET[\delta]$.
Thus, $\FunAlgErr{f_1}{\PRBQ} = \MaxErr{\PRBQ}$ and
$\FunAlgErr{f_2}{\PRB} = \MaxErrQ{\PRB}$.

Now, applying the first property from Definition \ref{def:approximating} to
$f_1\in \LCSET$ gives us a function $f_1'\in \QSET[\delta]$ such that 
$\Abs{\FunAlgErr{f'_1}{\PRBQ} -\MaxErr{\PRBQ}} \leq \frac{\delta}{2} \MaxErr{\PRBQ}$.
Since $\MaxErrQ{\PRBQ} \geq \FunAlgErr{f'_1}{\PRBQ}$,
we obtain that $\MaxErrQ{\PRBQ} \geq  \MaxErr{\PRBQ}(1 - \frac{\delta}{2})$.

Similarly, applying the second property from Definition \ref{def:approximating} to 
$f_2 \in \QSET[\delta]$, gives us a function $f_2' \in \LCSET$ with 
$\Abs{\FunAlgErr{f'_2}{\PRB} -\MaxErrQ{\PRB}} \leq \frac{\delta}{2} \MaxErr{\PRB}$.
Since $\MaxErr{\PRB} \geq \FunAlgErr{f'_2}{\PRB}$,
we have that $\MaxErr{\PRB} \geq  \MaxErrQ{\PRB} - \frac{\delta}{2}\MaxErr{\PRB}$,
or $\MaxErr{\PRB} \geq \frac{\MaxErrQ{\PRB}}{1 + \frac{\delta}{2}}$.
Also, by optimality of \PRBQ in \QSET[\delta],
$\MaxErrQ{\PRBQ} \leq  \MaxErrQ{\PRB}$.
Thus, we obtain that
$\MaxErr{\PRBQ} \leq
\frac{\MaxErrQ{\PRBQ}}{1 - \frac{\delta}{2}} \leq
\frac{\MaxErrQ{\PRB}}{1 - \frac{\delta}{2}} \leq
 \frac{\MaxErr{\PRB} (1 +  \frac{\delta}{2})}{1 - \frac{\delta}{2}}
\leq \MaxErr{\PRB} (1 + 2\delta).$
\end{proof}

In light of Theorem \ref{thm:approximate-LP}, it suffices to
exhibit classes of functions $\delta$-approximating \LCSET for which
the corresponding optimization problem can be solved efficiently. We do so for
metric spaces of bounded doubling dimension and metric spaces that are
contained on the line.

\subsection{A PTAS for Arbitrary Metric Spaces}
We first observe that since the error for any translation of a function \FE is the
same as for \FE, we can assume w.l.o.g.~that
$\Fe{\MEDP} = 0$ for all functions \FE considered in this section.
Thus, in this section, we implicitly restrict \LCSET 
to functions with $\Fe{\MEDP} = 0$.

We next describe a set \QSET[\delta] of functions which
$\delta$-approximate \LCSET. 
Roughly, we will discretize function values to different multiples of
\QUANT, and consider distance scales that are different multiples
of \QUANT. We later set 
$\QUANT = \frac{\delta}{48 \cdot 6^\DDIM + 6}$.
We then show in Lemma \ref{lem:timecomplexity} that \QSET[\delta] has size 
$n^{\log(2(1+\QUANT)/\QUANT)(2/\QUANT)^\DDIM} = n^{O(1)}$
for constant $\delta$; 
the discretized LP  can therefore be solved in time 
$O(\mathrm{poly}(n) \cdot n^{\log(2(1+\QUANT)/\QUANT)(2/\QUANT)^\DDIM})$
(using exhaustive search for the separation oracle), and we obtain
a PTAS for finding the optimum distribution for arbitrary metric
spaces.


We let $k=\log_2 \frac{1}{2\MEDV}$, and define a sequence of $k$
rings of exponentially decreasing diameter around \MEDP, that
divide the space into $k+1$ regions $\Ring{1},\ldots,\Ring{k+1}$.
Specifically, 
$\Ring{k+1} = \Set{x}{\Dist{x}{\MEDP} \leq 2\MEDV}$, and
$\Ring{i} = \Set{x}{2^{-i} < \Dist{x}{\MEDP} \leq 2^{-(i-1)}}$
for $i=1, \ldots, k$.
Notice that because $\MEDV \geq \frac{1}{2n}$, we have that $k \leq
\log n$ suffices to obtain a disjoint cover.

Since the metric space has doubling dimension \DDIM, each region \Ring{i}
can be covered with at most $(2/\QUANT)^\DDIM$ balls of diameter 
$2\QUANT \cdot \RDist{i}$.
Let \Gb{i}{j} denote the \Kth{j} ball from the cover of \Ring{i};
without loss of generality, each \Gb{i}{j} is non-empty
and contained in \Ring{i} (otherwise, consider its intersection with
\Ring{i} instead).
We call \Gb{i}{j} the \Kth{j} \todef{grid ball} for region $i$.
Thus, the grid balls cover all points, and there are at most
$(2/\QUANT)^\DDIM \cdot \log n$ grid balls.
See Figure \ref{fig:grid} for an illustration of this cover.

For each grid ball \Gb{i}{j}, let $\GbR{i}{j} \in \Gb{i}{j}$ be an
arbitrary, but fixed, \todef{representative} of \Gb{i}{j}.
The exception is that for the grid ball containing \MEDP, \MEDP must
be chosen as the representative.
We now define the class \QSET[\delta] of functions \FE as follows:
\begin{enumerate}
\item For each $i,j$, \Fe{\GbR{i}{j}} is a multiple of 
$\QUANT \cdot \RDist{i}$. 
\item For all $(i,j), (i',j')$, the function values satisfy the
following \todef{relaxed Lipschitz-condition}:
$\Abs{\Fe{\GbR{i}{j}} - \Fe{\GbR{i'}{j'}}} 
\leq \Dist{\GbR{i}{j}}{\GbR{i'}{j'}} + \QUANT \cdot (\RDist{i}+\RDist{i'})$.
\item All points in $\Gb{i}{j}$ have the same function value, i.e.,
$\Fe{x} = \Fe{\GbR{i}{j}}$ for all $x \in \Gb{i}{j}$.
\end{enumerate}

\begin{figure}[htb]
\begin{center}
\epsfxsize=10.5cm
\epsffile{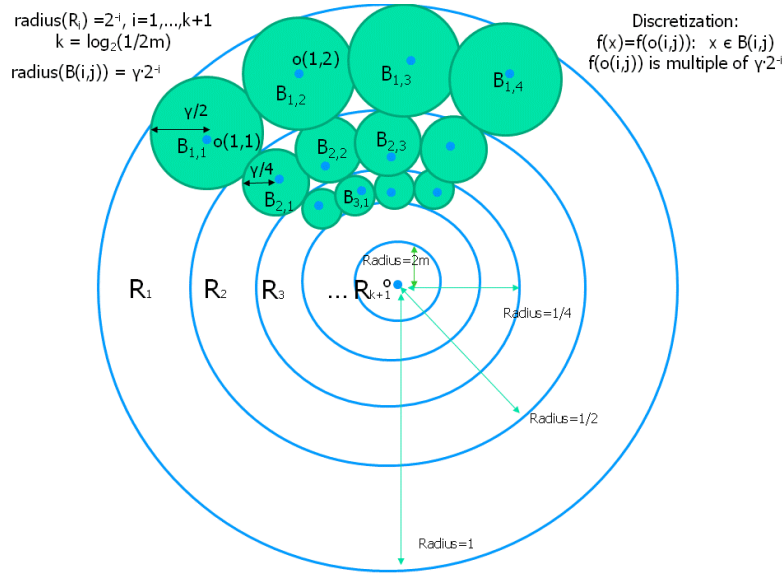}
\caption{Covering with grid balls \label{fig:grid}}
\end{center}
\end{figure}

We first show that the size of \QSET[\delta] is polynomial in $n$.

\begin{lemma} \label{lem:timecomplexity}
The size of \QSET[\delta] is at most
$n^{\log(2 (1+\QUANT)/\QUANT)(2/\QUANT)^\DDIM}$.
\end{lemma}

\begin{proof}
Because of the first and second constraint in the definition
of \QSET[\delta], each point \GbR{i}{j} can take on at most
$\frac{\Dist{\GbR{i}{j}}{\MEDP} + 2\QUANT \RDist{i}}{\QUANT \cdot 2^{-i}}
 \leq  \frac{2 (1+\QUANT) \cdot 2^{-i}}{\QUANT \cdot 2^{-i}}
 =  \frac{2(1+\QUANT)}{\QUANT}$
distinct values. 
Setting the values for all \GbR{i}{j} uniquely determines the function
\FE; the relaxed Lipschitz condition will result in some of these
functions not being in \QSET[\delta], and thus only decreases the
number of possible functions.
Because there are at most $(2/\QUANT)^\DDIM \cdot \log n$ grid balls,
there are at most 
$(2(1+\QUANT)/\QUANT)^{(2/\QUANT)^\DDIM \cdot \log n} 
= n^{\log(2(1+\QUANT)/\QUANT)(2/\QUANT)^\DDIM}$
functions in \QSET[\delta].
\end{proof}

We need to prove that \QSET[\delta] approximates \LCSET well, by
verifying that for each function $\FE \in \LCSET$, there is a
``close'' function in \QSET[\delta], and vice versa.
We first show that for any function satisfying the relaxed Lipschitz
condition, we can change the function values slightly and obtain a
Lipschitz continuous function.

\begin{lemma} \label{lem:lipschitz-adjust}
For each $x \in \METSPACE$, let $s_x$ be some non-negative number.
Assume that \FE satisfies the ``relaxed Lipschitz condition''
$\Abs{\Fe{x} - \Fe{y}} \leq \Dist{x}{y} + s_x + s_y$ for all $x,y$.
Then, there is a Lipschitz continuous function $\FEP \in \LCSET$ such that 
$\Abs{\Fe{x} - \FeP{x}} \leq s_x$ for all $x$.
\end{lemma} 

\begin{proof}
We describe an algorithm which runs in iterations $\ell$, and sets
the value of one point $x$ per iteration.
$S_\ell$ denotes the set of $x$ such that \FeP{x} has been set.
We maintain the following two invariants after the \Kth{\ell} iteration:
\begin{enumerate}
\item \FEP satisfies the Lipschitz condition for all pairs of points in $S_{\ell}$,
and $\Abs{\FeP{x} - \Fe{x}} \leq s_x$ for all $x \in S_{\ell}$.
\item For every function \FEPP satisfying the previous condition,
$\FeP{x} \leq \FePP{x}$ for all $x \in S_{\ell}$.
\end{enumerate}
Initially, this clearly holds for $S_0 = \emptyset$.
And clearly, if it holds after iteration $n$, the function \FEP
satisfies the claim of the lemma.

We now describe iteration $\ell$. 
For each $x \notin S_{\ell-1}$, let 
$t_x = \max_{y \in S_{\ell-1}} (\FeP{y} - \Dist{x}{y})$.
We show below that for all $x$, we have $t_x \leq \Fe{x} + s_x$.
Let $x \notin S_{\ell-1}$ be a point
maximizing $\max(\Fe{x}-s_x, t_x)$, and set
$\FeP{x} = \max(\Fe{x}-s_x, t_x)$. It is easy to verify
that this definition satisfies both parts of the invariant.

It remains to show that $t_x \leq \Fe{x} + s_x$ for all points 
$x \notin S_{\ell-1}$. Assume that $t_x > \Fe{x} + s_x$ for some point
$x$. Let $x_1$ be the point in $S_{\ell-1}$ for which 
$t_x = \FeP{x_1} - \Dist{x}{x_1}$.
By definition, either $\FeP{x_1} = \Fe{x_1}-s_{x_1}$, or there is an $x_2$
such that $\FeP{x_1} = t_{x_1} = \FeP{x_2} - \Dist{x_1}{x_2}$.
In this way, we obtain a chain $x_1, \ldots, x_r$ such that
$\FeP{x_i} = \FeP{x_{i+1}} - \Dist{x_i}{x_{i+1}}$ for all $i < r$,
and $\FeP{x_r} = \Fe{x_r} - s_{x_r}$.
Rearranging as $\FeP{x_{i+1}} - \FeP{x_i} = \Dist{x_i}{x_{i+1}}$,
and adding all these equalities for $i = 1, \ldots, r$ gives us that
$\Fe{x_r} - \FeP{x_1} = s_{x_r} + \sum_{i=1}^{r-1} \Dist{x_i}{x_{i+1}}$.
By assumption, we have
$\FeP{x_1} - \Dist{x}{x_1} = t_x > \Fe{x} + s_x$.
Substituting the previous equality, rearranging, and applying the
triangle inequality gives us that
\[ \begin{array}{lclcl}
\Fe{x_r} - \Fe{x}
& > & s_x + s_{x_r} + \Dist{x}{x_1} + \sum_{i=1}^{r-1} \Dist{x_i}{x_{i+1}}
& \geq & s_x + s_{x_r} + \Dist{x}{x_r},
\end{array} \]
which contradicts the relaxed Lipschitz condition for the pair
$x, x_r$. 
\end{proof}

We now use Lemma \ref{lem:lipschitz-adjust} to obtain,
for any given $\FE \in \QSET[\delta]$, a function $\FEP \in \LCSET$
close to \FE. 

\begin{lemma} \label{lem:LPdiff2}
Let $\FE \in \QSET[\delta]$. There exists an $\FEP \in \LCSET$ such
that $\Abs{\FunAlgErr{\FE}{\PRB} - 
\FunAlgErr{\FEP}{\PRB}}
\leq \frac{\delta}{2} \cdot \MaxErr{\PRB}$, for all distributions \PRB.
\end{lemma}

\begin{proof}
Because $\FE \in \QSET[\delta]$, it must satisfy,
for all $(i,j)$ and $(i',j')$, the relaxed Lipschitz condition
$\Abs{\Fe{\GbR{i}{j}} - \Fe{\GbR{i'}{j'}}} 
\leq \Dist{\GbR{i}{j}}{\GbR{i'}{j'}} + \QUANT \cdot (\RDist{i}+\RDist{i'})$.
Now, we apply Lemma \ref{lem:lipschitz-adjust} with
$s_{\GbR{i}{j}} = \QUANT \cdot \RDist{i}$ to get function values
\FeP{\GbR{i}{j}} for all $i,j$ that satisfy the Lipschitz condition
and the condition that 
$\Abs{\FeP{\GbR{i}{j}} - \Fe{\GbR{i}{j}}} \leq \QUANT \cdot \RDist{i}$.
For any other point $x$, let
$\MaxVal{x}{\FE} = \min_{i,j}(\FeP{\GbR{i}{j}} + \Dist{x}{\GbR{i}{j}})$ and 
$\MinVal{x}{\FE} = \max_{i,j}(\FeP{\GbR{i}{j}} - \Dist{x}{\GbR{i}{j}})$, and set
$\FeP{x} = \half \cdot (\MaxVal{x}{\FE} + \MinVal{x}{\FE})$.
It is easy to see that $\MinVal{x}{\FE} \leq \MaxVal{x}{\FE}$ for all
$x$, and that this definition gives a Lipschitz continuous function \FEP.
For a point $x \in \Gb{i}{j}$, triangle inequality, the above
construction, and the fact that \Gb{i}{j} has diameter 
at most $2\QUANT \cdot \RDist{i}$ imply that
\begin{eqnarray*}
\Abs{\FeP{x} - \Fe{x}} 
& \leq & \Abs{\FeP{x} - \FeP{\GbR{i}{j}}} 
+ \Abs{\FeP{\GbR{i}{j}} - \Fe{\GbR{i}{j}}}
+ \Abs{\Fe{\GbR{i}{j}} - \Fe{x}}\\
& \leq & 2\QUANT \cdot \RDist{i}
+ \QUANT \cdot \RDist{i}
+ 0  \\
& = & 3\QUANT \cdot \RDist{i}.
\end{eqnarray*}

For each point $x$, let \PRing{x} be the index of the region $i$ such
that $x \in \Ring{i}$. Now, using the triangle inequality and
Lemma \ref{lem:lowerboundkd}, we can bound
\begin{eqnarray*}
\Abs{\Avg{\FEP} - \Avg{\FE}}
& \leq & \frac{1}{n} \cdot \sum_x \Abs{\FeP{x} - \Fe{x}} \\
& \leq & \frac{1}{n} \cdot \sum_x 3\QUANT \cdot \RDist{\PRing{x}} \\
& \leq & \frac{1}{n} \cdot (\sum_{x \notin R_{k+1}} 3\QUANT \cdot \Dist{x}{\MEDP}
+ \sum_{x \in R_{k+1}} 3\QUANT \cdot \MEDV) \\
& \leq & \frac{1}{n} \cdot (3 \QUANT n \MEDV + 3 \QUANT n \MEDV)\\
& \leq & 24 \cdot 6^\DDIM \cdot \QUANT \cdot \MaxErr{\PRB}.
\end{eqnarray*}

Similarly, we can bound
\begin{eqnarray*}
\sum_x \Prb{x} \cdot \Abs{\FeP{x} - \Fe{x}}
& \leq &
3\QUANT \cdot (\sum_{x \notin R_{k+1}} \Prb{x} \cdot \Dist{x}{\MEDP}
+ \sum_{x \in R_{k+1}} \Prb{x} \MEDV)\\
& \leq &
3\QUANT \cdot (\MEDV + \sum_{x \notin R_{k+1}} \Prb{x} \cdot \Dist{x}{\MEDP}).
\end{eqnarray*}

Let \FEPP be defined as $\FePP{x} = \Dist{x}{\MEDP}$. Clearly, 
$\FEPP \in \LCSET$, $\Avg{\FEPP} = \MEDV$,
and the estimation error for \PRB when the input is \FEPP is
\[ \begin{array}{lclclcl}
\FunAlgErr{\FEPP}{\PRB}
& = & \sum_x \Prb{x} \cdot \Abs{\FePP{x} - \MEDV}
& \geq &  
\sum_{x \notin \Ring{k+1}} \Prb{x} \cdot \Abs{\Dist{x}{\MEDP} - \MEDV}
& \geq &
(\sum_{x \notin \Ring{k+1}} \Prb{x} \cdot \Dist{x}{\MEDP}) - \MEDV.
\end{array} \]

Combining these observations, and using Lemma \ref{lem:lowerboundkd} and the fact that
$\FunAlgErr{\FEPP}{\PRB} \leq \MaxErr{\PRB}$, we get
$\sum_x \Prb{x} \cdot \Abs{\FeP{x} - \Fe{x}}
 \leq  6\QUANT \cdot \MEDV
+ 3\QUANT \FunAlgErr{\FEPP}{\PRB}
 \leq 
(8\cdot 6^\DDIM + 1) \cdot 3\QUANT \cdot \MaxErr{\PRB}$.

Now, by using the fact that $\Abs{\FunAlgErr{\FE}{\PRB} - 
\FunAlgErr{\FEP}{\PRB}} \leq \Abs{\Avg{\FEP} - \Avg{\FE}} + 
\sum_x \Prb{x} \cdot \Abs{\FeP{x} - \Fe{x}}$, and 
setting $\QUANT = \frac{\delta}{48 \cdot 6^\DDIM + 6}$, we 
obtain the desired bound.
\end{proof}

Finally, we need to analyze the converse direction.

\begin{lemma} \label{lem:LPdiff1}
Let $\FE \in \LCSET$. There exists an $\FEP \in \QSET[\delta]$ such that
$\Abs{\FunAlgErr{\FE}{\PRB} - 
\FunAlgErr{\FEP}{\PRB}}
\leq \frac{\delta}{2} \cdot \MaxErr{\PRB}$, for all distributions \PRB.
\end{lemma}

\begin{proof}
The proof is similar to that of Lemma \ref{lem:LPdiff2}.
First, for each grid ball representative \GbR{i}{j}, we
let \FeP{\GbR{i}{j}} be \Fe{\GbR{i}{j}}, rounded down (up for negative
numbers) to the nearest multiple of $\QUANT \cdot \RDist{i}$.
Then, for all points $x \in \Gb{i}{j}$, we set
$\FeP{x} = \FeP{\GbR{i}{j}}$. Clearly, the resulting function
\FEP is in \QSET[\delta].

By a similar argument as before, 
$\Abs{\FeP{x} - \Fe{x}} \leq 3\QUANT \cdot \RDist{\PRing{x}}$, 
for all points $x$.
Thus, we immediately get
$\Abs{\Avg{\FEP} - \Avg{\FE}} \leq
24 \cdot 6^\DDIM \cdot \QUANT \cdot \MaxErr{\PRB}$ as well.

Define the function \FEPP exactly as in the proof of Lemma \ref{lem:LPdiff2}.
Then, exactly the same bounds as in that proof apply, and give us the
claim.
\end{proof}
%

\subsection{An FPTAS for points on a line} \label{sec:FPTAS}
In this section, we show that if the metric consists of a discrete point set 
on the line, then the general PTAS of the previous section can be improved to 
an FPTAS. 

Since we 
assumed the maximum distance to be $1$, we can assume w.l.o.g.~that the points
are $0= x_1 \leq x_2 \leq \cdots \leq x_n = 1$.
Also, because w.l.o.g. the post-processing is the identity function, 
we only need to consider functions 
$\FE \in \LCSET[0]$, i.e., such that $\sum_i \Fe{x_i} = 0$.
We define $\QUANT = \frac{\delta}{144n}$, and the class \QSET[\delta]
to contain the following functions \FE:
\begin{enumerate}
\item For each $i$, \Fe{x_i} is a multiple of \QUANT.
\item The function values satisfy the relaxed Lipschitz-condition
$\Abs{\Fe{x_i} - \Fe{x_j}} \leq \Dist{x_i}{x_j} + \QUANT$ for all $i,j$.
\item The sum is ``close to 0'', in the sense that 
$\sum_i \Fe{x_i} \leq n \QUANT$.
\end{enumerate}

We first establish that, given a probability distribution \PRB, a
function $\FE \in \QSET[\delta]$ maximizing 
$\sum_i \Prb{x_i} \cdot \Abs{\Fe{x_i}}$
can be found in polynomial time using Dynamic Programming. 
To set up the recurrence, let 
\DPErr{j}{t}{s} be the maximum expected error that can be achieved with
function values at $x_1, \ldots, x_j$, under the constraints that
$\Fe{x_j} = t$ and $\sum_{i \leq j} \Fe{x_i} = s$.
Then, we obtain the recurrence

\[ \begin{array}{lcl}
\DPErr{1}{t}{s}
& = & \CaseDist{\Prb{x_1} \cdot t}{\mbox{ if }s=t}{-\infty}\\
\DPErr{j+1}{t}{s}
& = & \max_{y \in [t-(x_{j+1}-x_j), t+(x_{j+1}-x_j)], \QUANT | y}
(\Prb{x_{j+1}} \Abs{t} + \DPErr{j}{y}{s-t})
\end{array} \]

The maximizing value is then 
$\max_{s \in [-n\QUANT, n\QUANT], \QUANT | s, t \in [-1, 1], \QUANT | t}
\DPErr{n}{t}{s}$.
The total number of entries is $O(n \cdot \frac{1}{\QUANT^2})$, and each
entry requires time $O(\frac{1}{\QUANT})$ to compute.
The overall running time is thus 
$O(n \cdot \frac{1}{\gamma^3}) = O(\frac{n^4}{\delta^3})$, 
giving us an FPTAS.

All we need now is to show that $\QSET[\delta]$ 
$\delta$-approximates \LCSET. We use the following lemma:

\begin{lemma} \label{lem:LPdiff1d}
For each $\FE \in \LCSET$, there is a function
$\FEP \in \QSET[\delta]$ such that
$\Abs{\FunAlgErr{\FE}{\PRB} -
\FunAlgErr{\FEP}{\PRB}}
\leq 2\QUANT$ for all distributions \PRB.
Also, for each $\FE \in \QSET[\delta]$, there is a function
$\FEP \in \LCSET$ such that
$\Abs{\FunAlgErr{\FE}{\PRB} -
\FunAlgErr{\FEP}{\PRB}}
\leq 6\QUANT$ for all distributions \PRB.
\end{lemma}

\begin{proof}
For the first part, define \FEP by rounding each \Fe{x_i} down 
(toward 0 for negative values) to the nearest multiple of \QUANT.
Clearly, $\FEP \in \QSET[\delta]$. 
Furthermore, the average changes by at most \QUANT,
and $\sum_x \Prb{x} \Abs{\FeP{x} - \Fe{x}} \leq \sum_x \Prb{x} \QUANT = \QUANT$.

For the second part, first create a Lipschitz continuous function
\FEPP from \FE according to Lemma \ref{lem:lipschitz-adjust}; 
then define $\FeP{x} = \FePP{x} - \Avg{\FEPP}$ for all $x$.
The first step changed each function value by at most \QUANT,
and because $\Avg{\FEPP} \leq \QUANT + \Avg{\FE} \leq 2\QUANT$.
we have that $\Abs{\FeP{x} - \Fe{x}} \leq 3\QUANT$ for all $x$.
Thus, $\Abs{\Avg{\FEP} - \Avg{\FE}} + \sum_x \Prb{x} \cdot \Abs{\FeP{x} - \Fe{x}} 
\leq 6\QUANT$.
\end{proof}

By Lemma \ref{lem:lowerboundkd}, applied with $\DDIM = 1$, any
randomized algorithm must have expected error at least 
$\frac{1}{12n}$.
In particular, substituting the definition of 
$\QUANT = \frac{\delta}{144n}$ gives us that
$6\QUANT \leq \frac{\delta}{2} \cdot \MaxErr{\PRB}$ for all
distributions \PRB. Thus, \QSET[\delta] approximates \LCSET well.

\section{Sampling in the Interval $[0,1]$} \label{sec:continuous}
In this section, we focus on what is probably the most basic version
of the problem: the metric space is the interval $[0,1]$. 
In this continuous case, we can explicitly characterize the optimum
sampling distribution and estimation error.
It is easy to see (and follows from a more general result in
\cite{das}) that the best deterministic algorithm samples the function
at \half and outputs the value read. 
The worst-case error of this algorithm is \quarter.
We prove that randomization can lead to the following improvement.

\begin{theorem} \label{thm:optimal-sampling}
An optimal distribution that minimizes the worst-case expected 
estimation error is to sample uniformly from the interval
$[2-\sqrt{3}, \sqrt{3}-1]$. 
This sampling gives a worst-case error of 
$1-\frac{\sqrt{3}}{2} \approx 0.134$.
\end{theorem}

Following the discussion in Section \ref{sec:preliminaries}, we
restrict our analysis w.l.o.g.~to functions $\FE \in \LCSET[0]$, i.e.,
we assume that $\int_0^1 \Fe{x} dx = 0$.
Then, the expected error of a distribution \PRB against input \FE is
$\FunAlgErr{\FE}{\PRB} = \int_0^1 \Prb{x} \Abs{\Fe{x}} dx$.
The key part of the proof of Theorem \ref{thm:optimal-sampling} is to
show that when the algorithm samples uniformly over an interval 
$[c,1-c]$, then with loss of only an arbitrarily small $\epsilon$, 
we can focus on functions consisting of just two line segments.


\begin{theorem} \label{thm:worstcasefunction}
For any \PEAK, define
$\Fe[\PEAK]{x} = \half + \PEAK^2 - \PEAK - \Abs{\PEAK-x}$.
If \PRB is uniform over $[c,1-c]$ where $c=2-\sqrt{3}$, then for every $\epsilon > 0$, 
there exists some $\PEAK=\PEAK(\epsilon)$ such that for all 
functions $\FE \in \LCSET[0]$, we have 
$\FunAlgErr{\FE[\PEAK]}{\PRB} \geq \FunAlgErr{\FE}{\PRB} - \epsilon$.
\end{theorem}

All of Section \ref{sec:worst-case-characterization} is devoted to the
proof of Theorem \ref{thm:worstcasefunction}.
Here, we show how to use Theorem \ref{thm:worstcasefunction} to prove
the upper bound from Theorem \ref{thm:optimal-sampling}.

Let $c=2-\sqrt{3}$, so that the algorithm samples uniformly from $[c,1-c]$.
Let $\epsilon > 0$ be arbitrary; we later let $\epsilon \to 0$.
Let $\PEAK=\PEAK(\epsilon)$ be the value whose existence is guaranteed
by Theorem \ref{thm:worstcasefunction}.
We distinguish two cases:
\begin{enumerate}
\item If $\PEAK \leq c$, then
\begin{eqnarray*}
\FunAlgErr{\FE[\PEAK]}{\PRB}
& = &
\frac{1}{1-2c} \cdot \int_c^{1-c} \Abs{\half + \PEAK^2 - \PEAK - \Abs{\PEAK-x}} dx  \\
& = &
\frac{1}{1-2c} \cdot (\half(\PEAK^2 + \half - c)^2 + \half(1-c - \PEAK^2)^2) \\
& = &
\frac{1}{1-2c} \cdot (\PEAK^4 + (\half - c)^2).
\end{eqnarray*}

\item If $\PEAK \geq c$, then
\begin{eqnarray*}
\FunAlgErr{\FE[\PEAK]}{\PRB}
& = &
\frac{1}{1-2c} \cdot \int_c^{1-c} \Abs{\half + \PEAK^2 - \PEAK - \Abs{\PEAK-x}} dx  \\
& = &
\frac{1}{1-2c} \cdot (2\PEAK \Fe{\PEAK} + 2c\PEAK + \Fe{\PEAK}^2 - c - \PEAK^2)\\
& = &
\frac{1}{1-2c} (\PEAK^4 - \PEAK^2 + 2c\PEAK + \quarter - c) \\
& = &
\frac{1}{1-2c} (\PEAK^4 - (\PEAK-c)^2 + (\half - c)^2).
\end{eqnarray*}
\end{enumerate}
The first formula is increasing in \PEAK, and thus maximized at $\PEAK=c$;
at $\PEAK = c$, the value equals that of the second formula, so the
maximization must happen for $\PEAK \geq c$. A derivative test shows
that it is maximized for $\PEAK = \frac{\sqrt{3}-1}{2}$,
giving an error of $1-\frac{\sqrt{3}}{2}$.
By Theorem \ref{thm:worstcasefunction}, for any function \FE, the
error is at most $1-\frac{\sqrt{3}}{2} + \epsilon$, and letting
$\epsilon \to 0$ now proves an upper bound of 
$1-\frac{\sqrt{3}}{2}$ on the error of the given distribution.

Next, we prove optimality of the uniform distribution over
$[2-\sqrt{3}, \sqrt{3}-1]$, by providing a lower bound on all
randomized sampling distributions. 
Again, by Theorem \ref{thm:randpredictionfunction}, we focus only
on algorithms which output the value \Fe{x} after sampling at $x$, by
incurring an error $\epsilon > 0$ that can be made arbitrarily small.
Our proof is based on Yao's Minimax principle: 
we explicitly prescribe a distribution \DISTR over \LCSET[0] such that 
for any deterministic algorithm using the identity function, the
expected estimation error is at least $1 - \frac{\sqrt{3}}{2}$.
Since a deterministic algorithm is characterized completely by its
sampling location $x$, this is equivalent to showing that 
$\Expect[\FE \sim \DISTR]{\Abs{\Fe{x}}} \geq 1 - \frac{\sqrt{3}}{2}$ for all $x$.

We let $\PEAK = \frac{\sqrt{3}-1}{2}$, and define two functions 
$\FEA, \FEB$ as $\FeA{x} = \half+\PEAK^2 -\PEAK - \Abs{x-\PEAK}$
and $\FeB{x}= \FeA{1-x}$. The distribution \DISTR is then simply to
choose each of \FEA and \FEB with probability \half.
Fix a sampling location $x$; by symmetry, we can restrict ourselves to
$x \leq \half$.
Because $\Avg{\FEA} = \Avg{\FEB} = 0$, the expected estimation error is 
\begin{eqnarray*}
\half (\Abs{\FeA{x}} + \Abs{\FeB{x}})
& = & \half (  \Abs{\half + \PEAK^2 - \PEAK - \Abs{x-\PEAK}} 
            + \Abs{\half + \PEAK^2 - \PEAK - \Abs{1-x-\PEAK}})\\
& = &
\left\{ \begin{array}{ll}
\half - \PEAK, & \mbox{ if } x \leq \PEAK \\
\half - x, & \mbox{ if } \PEAK \leq x \leq \half - \PEAK^2 \\
\PEAK^2, & \mbox{ if } \half - \PEAK^2 \leq x \leq \half.
\end{array} \right.
\end{eqnarray*}

This function is clearly non-increasing in $x$, and thus minimized at
$x=\half$, where its value is $\PEAK^2 = 1-\frac{\sqrt{3}}{2}$.
Thus, even at the best sampling location $x=\half$, the error cannot
be less than $1-\frac{\sqrt{3}}{2}$.
This completes the proof of Theorem \ref{thm:optimal-sampling}.\QED

Notice that the proof of Theorem \ref{thm:optimal-sampling} has an
interesting alternative interpretation. For a (finite) multiset 
$S \subset \LCSET[0]$ of Lipschitz continuous functions \FE with
 $\int_x \Fe{x} dx = 0$, we say that $S$ is 
\todef{$\delta$-close} if there exist $x,y$ such that
$\frac{1}{n} \cdot \sum_{\FE \in S} \Abs{\Fe{x}-y} \leq \delta$.
In other words, the average distance of the functions from a carefully
chosen reference point is at most $\delta$. Then, the proof of Theorem
\ref{thm:optimal-sampling} implies:

\begin{theorem} \label{thm:close}
Every set $S \subseteq \LCSET[0]$ is $(1-\frac{\sqrt{3}}{2})$-close, and
this is tight.
\end{theorem}

\subsection{Proof of Theorem \ref{thm:worstcasefunction}}
\label{sec:worst-case-characterization}

We begin with the following lemma which
guarantees that we can focus on functions \FE with finitely many
zeroes.

\begin{lemma} \label{lem:worstcasefinitezeros}
For any $\epsilon > 0$ and any function \FE, there exists a function
\FEP such that there are at most $O(1/\epsilon)$ points $x$ with
$\FeP{x} = 0$, and 
$\FunAlgErr{\FEP}{\PRB} \geq \FunAlgErr{\FE}{\PRB} - \epsilon$, 
for all distributions \PRB.
\end{lemma}

\begin{proof}
Let $\epsilon > 0$ be arbitrary. 
We prove the lemma by modifying
\FE to ensure that it meets the requirements, and
showing that its estimation error decreases by at most $\epsilon$ in
the process.

We replace \FE with a function \FEP with the following properties:
(1) \FEP is Lipschitz continuous,
(2) $\int_0^1 \Fe{x} dx = \int_0^1 \FeP{x} dx$, 
(3) $\Abs{\Fe{x} - \FeP{x}} \leq \epsilon$ for all $x$, and
(4) for each $j = 1, \ldots, 1/\epsilon$, the set
$\ZS{\FEP}{j} = \Set{x \in [(j-1)\epsilon, j\epsilon]}{\FeP{x} = 0}$ 
contains at most three points. 
The error can change by at most $\epsilon$ due to the third condition,
and the fourth condition ensures the bound on the number of zeroes.

To describe the construction, first focus on one interval $[(j-1)\epsilon, j\epsilon]$,
and define 
$x^- = \inf \ZS{\FE}{j}$,  $x^+ = \sup \ZS{\FE}{j}$, and $\delta = x^+-x^-$. 
Now let $\alpha = \frac{\delta^2 + 4\int_{x^-}^{x^+} \Fe{x} dx}{4\delta}$,
and define the function 
$\FEP$ such that
\[ \begin{array}{lcl}
\FeP{x}
& = &
\left\{ \begin{array}{ll}
\alpha-\Abs{x^-+\alpha-x}, & \mbox{ if } x \in [x^-, x^-+2\alpha]\\
\alpha-\delta/2+\Abs{x^+ +\alpha - \delta/2-x}, & \mbox{ if } x \in [x^-+2\alpha, x^+] \\
\Fe{x}, & \mbox{ if } x \in [(j-1)\epsilon, j\epsilon] \setminus [x^-, x^+].
\end{array} \right.
\end{array} \]
Intuitively, this replaces the function on the interval by a zigzag
shape with the same integral that has the same leftmost and rightmost
zero. 

Do this for each $j$. By the careful choice of $\alpha$, the integral
remains unchanged. Because each function value changes by at most
$\delta \leq \epsilon$, the third condition is satisfied; the fourth
condition is directly by construction, and Lipschitz continuity is
obvious.
\end{proof}

Next, we show a series of lemmas restricting the functions \FE under
consideration. When we say that \FE has a certain property without
loss of generality, we mean that changing \FE to \FEP with that
property can be accomplished while ensuring that
$\FunAlgErr{\FEP}{\PRB} \geq \FunAlgErr{\FE}{\PRB}$ for all
uniform distributions \PRB over intervals $[c,1-c]$.
Since our goal is to characterize the functions that make the
algorithm's error large, this restriction is indeed without loss of
generality.

We focus on points $x \in (c, 1-c)$ with $\Fe{x} = 0$.
Let $c \leq z_1 \leq \ldots \leq z_k \leq 1-c$ be all such points.
For ease of notation, we write $z_0 = c$ and $z_{k+1} = 1-c$.
By continuity, \Fe{x} has the same sign for all $x \in (z_i, z_{i+1})$,
for $i=0, \ldots, k$.
We show that w.l.o.g., \FE is as large as possible over areas of
the same sign.

\begin{lemma} \label{lem:worstcasetriangles1}
Assume w.l.o.g.~that $\Fe{x} \geq 0$ for all $x \in [z_i, z_j]$, with $j > i$.
Then, w.l.o.g., \FE maximizes the area over $[z_i, z_j]$ subject to the
Lipschitz constraint and the function values at $z_i$ and $z_j$. More
formally, w.l.o.g., \FE satisfies,
\begin{enumerate}
\item If $1 \leq i < j \leq k$, then $\Fe{x} = \min(x-z_i, z_j-x)$ for all $x \in [z_i, z_j]$.
\item If $i=0$, then $\Fe{x} = \min(\Fe{c} + (x-c), z_1 - x)$ for all $x \in [c, z_1]$, and
if $i=k$, then $\Fe{x} = \min(\Fe{1-c} + (1-c)-x, x - z_k)$ for all $x \in [z_k, 1-c]$.
\end{enumerate}
\end{lemma}

\begin{proof}
We prove the first part here (the proof of the second part is analogous). 
Define a function \FEP as
$\FeP{x} =  \min(x-z_i, z_j-x)$ for $x \in[z_i,z_j]$, and
$\FeP{x} = \Fe{x}$ otherwise.
Let $\FEPP = \FEP - \Avg{\FEP}$, so that \FEPP is renormalized to
have integral 0.
Since $\FeP{x} \geq \Fe{x}$ for all $x$, and $\Avg{\FE} = 0$, we
have that $\Avg{\FEP} \geq 0$. Then
\begin{eqnarray*}
& & \int_c^{1-c} \Abs{\FePP{x}} - \Abs{\Fe{x}} dx\\
& = &\int_{z_i}^{z_j} \Abs{\FePP{x}} - \Abs{\Fe{x}} dx + \int_c^{z_i} \Abs{\Fe{x} - \Avg{\FEP}} - \Abs{\Fe{x}} dx +
  \int_{z_j}^{1-c} \Abs{\Fe{x} - \Avg{\FEP}} - \Abs{\Fe{x}} dx \\
& \geq & \int_{z_i}^{z_j} (\Abs{\FeP{x} - \Avg{\FEP}} - \Abs{\FeP{x}})
+ (\Abs{\FeP{x}} - \Abs{\Fe{x}}) dx  - (1 - 2c - (z_j-z_i)) \Avg{\FEP} \\
& \geq & \int_{z_i}^{z_j} \Abs{\FeP{x}} - \Abs{\Fe{x}} dx -
\int_{z_i}^{z_j} \Avg{\FEP} dx
 - (1- 2c - (z_j-z_i)) \Avg{\FEP} \\
& = & \int_{z_i}^{z_j} \FeP{x} - \Fe{x} dx - (1-2c) \Avg{\FEP} \\
& = & 2c \cdot \Avg{\FEP} \\
& \geq & 0.
\end{eqnarray*}

Thus, the estimation error of \FEPP is at least as large as the one for
\FE, so w.l.o.g., \FE satisfies the statement of the lemma.
\end{proof}

\begin{lemma} \label{lem:worstcasetriangles3}
W.l.o.g., $k \leq 2$, i.e., there are at most two
points $x \in (c, 1-c)$ such that $\Fe{x} = 0$.
\end{lemma}

\begin{proof}
Assume that $\Fe{z_1} = \Fe{z_2} = \Fe{z_3} = 0$.
Consider mirroring the function on the interval $[z_1,z_3]$.
Formally, we define $\FeP{x} = \Fe{z_3-(x-z_1)}$ 
if $x \in [z_1, z_3]$,
and $\FeP{x} = \Fe{x}$ otherwise.

Clearly, \FEP is Lipschitz continuous and has the same
average and same expected estimation error as \FE.
However, the signs of \FEP
on the intervals $[c,z_1]$ and $[z_1, z_1+z_3-z_2]$ are now the same;
similarly for the intervals $[z_1+z_3-z_2, z_3]$ and $[z_3,1-c]$.
Thus, applying Lemma \ref{lem:worstcasetriangles1} , we can further reduce the 
number of points $x$ with $\Fe{x} = 0$, without decreasing the estimation
error. 
\end{proof}

Hence, it suffices to focus on functions \FE that have at most
two points $z \in (c, 1-c)$ with $\Fe{z} = 0$. 
We distinguish three cases accordingly:

\begin{enumerate}
\item If there is no point $z \in (c,1-c)$ with 
$\Fe{z} = 0$, then \Fe{c} and \Fe{1-c} have the same sign; 
without loss of generality, \FE is negative over $(c,1-c)$. 
Then, the expected error is maximized when 
$\int_0^{c} \Fe{x} dx$ and $\int_{1-c}^{1} \Fe{x} dx$ are as positive as possible, 
subject to the Lipschitz condition and the constraint that $\int_0^1 \Fe{x} dx = 0$. 
Otherwise, we could increase the value of $\int_0^{c} \Fe{x} dx$ and $\int_{1-c}^{1} \Fe{x} dx$, and then lower the function to restore the integral to 0. By doing this, the
expected estimation error cannot decrease. 
Thus, by Lemma \ref{lem:worstcasetriangles1}, \FE is of the form
$\Fe{x}= \Abs{x-b} + \Fe{b}$, where $b = \argmin_{x \in (c,1-c)} \Fe{x}$.

\item If there is exactly one point $z \in (c,1-c)$ with 
$\Fe{z} = 0$, then \Fe{c} and \Fe{1-c} have opposite signs. Without
loss of generality, assume that $\Fe{c} > 0 > \Fe{1-c}$
and that $z \leq \half$. 
(Otherwise, we could consider $\FeP{x} = \Fe{1-x}$ instead.)

The expected error is maximized when \Fe{c} is as large as possible,
and $\int_z^{1-c} \Fe{x} dx$ is as negative as possible, subject to
the Lipschitz condition and the constraint that $\int_0^1 \Fe{x} dx = 0$. 
Because $z \leq \half$ and the integral of the function
$\FeP{x} = z-x$ is thus negative, by starting from \FEP, then raising the
function in the interval $[1-c,1]$ and, if necessary, increasing \FeP{1-c},
it is always possible to ensure that $\Fe{x} = z-x$
for all $x \in [0,z]$. Then, $\int_z^{1-c} \Fe{x} dx$ is as negative
as possible if, for some value $b$, \FE is of the following form:
$\Fe{x} = -(x-z)$ for $x \leq b$,
and $\Fe{x} = -(b-z) + (x-b) = z+x-2b$ for $x \geq b$.
Thus, \FE overall is of the form $\Fe{x} = \Abs{x-b} - (b-z)$.

\item 
If there are two points $z_1 < z_2 \in (c,1-c)$ with 
$\Fe{z_1} = \Fe{z_2} = 0$, then \Fe{c} and \Fe{1-c} have the same
sign; w.l.o.g., they are both positive.
We distinguish two subcases: 

\begin{itemize}
\item
If $z_2-z_1 \geq (1-c-z_2) + (z_1-c)$,
then the expected error is maximized when $\int_0^{z_1} \Fe{x} dx$ and
$\int_{z_2}^{1} \Fe{x} dx$ are as positive as possible, subject to the 
Lipschitz condition and the constraint that $\int_0^1 \Fe{x} dx = 0$.
Otherwise, we could increase the value of $\int_0^{z_1} \Fe{x} dx$ and 
$\int_{z_2}^{1} \Fe{x} dx$, and then lower the function by some small
resulting $\delta$ to restore the integral to 0. 
If the function is thus lowered by $\delta$, then for the
interval $(z_1,z_2)$, the error increases by $\delta$, while for the
intervals $[c,z_1]$ and $[z_2,1-c]$, it at most decreases by $\delta$.
By the condition $z_2-z_1 \geq (1-c-z_2) + (z_1-c)$, the lowering
would overall increase the error.
Now, applying Lemma \ref{lem:worstcasetriangles1} gives us that w.l.o.g.,
$\Fe{x} = \Abs{x-\frac{z_1+z_2}{2}} - \frac{z_2-z_1}{2}$.
 
\item
If $z_2-z_1 < (1-c-z_2) + (z_1-c)$, then the expected error is maximized 
when $\int_0^{c} \Fe{x} dx$ and
$\int_{1-c}^{1} \Fe{x} dx$ are as negative as possible, subject to the
Lipschitz condition and the constraint that $\int_0^1 \Fe{x} dx = 0$. 
Otherwise, we could decrease the value of $\int_0^{c} \Fe{x} dx$ and
$\int_{1-c}^{1} \Fe{x} dx$, and then raise the function to restore the
integral to 0. An argument just as in the previous case shows
  that the error cannot decrease.
Hence, w.l.o.g., $\Fe{x} = \Fe{c}-(c-x)$ for $x \in [0,c]$ and  
$\Fe{x} = \Fe{1-c}-(x-(1-c))$ for $x \in [1-c,1]$.

We next claim that there must be at least one point $z \in [0,c) \cup (1-c,1]$ 
such that $\Fe{z}=0$. 
For contradiction, assume that \FE is positive in $[0,c) \cup (1-c,1]$. 
Then, $\Fe{0}, \Fe{1} > 0$, and therefore, $\Fe{c}, \Fe{1-c} > c$.
Because $\Fe{z_1}=\Fe{z_2}=0$, this implies that $z_1 > 2c$ and
$z_2 < 1-2c$. But with our choice of $c=2-\sqrt{3}$, this implies that
$z_2 < z_1$, a contradiction. 

Without loss of generality, assume that the interval $(1-c,1]$
contains such a point $z$; define $z_3 = \min \Set{z \in (1-c,1]}{\Fe{z}=0}$. Further, assume that we have applied 
Lemma \ref{lem:worstcasetriangles1} to \FE, such that \FE
maximizes the area in the intervals $[c,z_1], [z_1,z_2]$ and 
$[z_2,1-c]$. 
Consider mirroring the function \FE on the interval $[z_1,z_3]$.
Formally, we define $\FeP{x} = \Fe{z_3-(x-z_1)}$ if $x \in [z_1, z_3]$,
and $\FeP{x} = \Fe{x}$ otherwise. 
Clearly, \FEP is Lipschitz continuous and has the same
integral (namely, zero) as \FE. 

Next, we define a new
function \FEPP by modifying \FEP so that it is as negative as 
possible in the interval $[z_1 + z_3 -z_2,1]$. Formally, we define 
$\FePP{x} = z_1+z_3-z_2-x$ if $x \in [z_1+z_3-z_2, 1]$,
and $\FePP{x} = \FeP{x}$ otherwise. 
(See Figure \ref{fig:mirroring} for an illustration of this 
mirroring, and the resulting shapes of \FEP and \FEPP)

Notice that \FEPP is not normalized to have an 
integral of $0$, since $\int_{z_1+z_3-z_2}^{1} \FePP{x} 
< \int_{z_1+z_3-z_2}^{1} \FeP{x}$. 
However, since (by assumption on the current case)) 
$z_2-z_1 < (1-c-z_2) + (z_1-c)$, 
raising \FEPP to restore the integral to $0$ can only
increase the resulting estimation error, by an argument
  similar to the previous case. The remainder of the 
proof for this case is as follows: We will first prove that the
estimation error of \FEPP is at least as large as the estimation
error of \FE. This implies that even after normalizing 
\FEPP, its estimation error remains at least as large as that 
of \FE. Finally we can use Lemma
\ref{lem:worstcasetriangles1} on the normalized version of \FEPP
to reduce the number of points $x \in (c,1-c)$ with $\FePP{x} = 0$
down to either one or zero, without decreasing the estimation error,
and thus reduce this subcase to one of the previous two cases.

We now compare the estimation error of \FE
against that of \FEPP. 
Simply by definition of \FEPP, we have that
$\int_{z_2}^{1-c} \Abs{\Fe{x}} dx =\int_{z_1}^{z_1+(1-c-z_2)} \Abs{\FePP{x}} dx$,
and $\int_{c}^{z_1} \Abs{\Fe{x}} dx =\int_{c}^{z_1} \Abs{\FePP{x}} dx$.
Furthermore, we have that $\int_{z_1}^{z_2} \Abs{\Fe{x}} dx \leq
\int_{z_1+(1-c-z_2)}^{1-c} \Abs{\FePP{x}} dx$. This follows, since
for any values $p,q$ such that $0 < p < q$, we have $\int_{0}^{q} \Abs{\frac{q}{2} -
\Abs{x-\frac{q}{2}}} dx \leq  \int_{0}^{q} \Abs{p-x} dx$.  
Hence, $\int_{c}^{1-c} \Abs{\FePP{x}} dx \geq
\int_{c}^{1-c} \Abs{\Fe{x}} dx$. 
\end{itemize}

\end{enumerate}

In all three cases, we have thus shown that w.l.o.g., 
$\Fe{x} = \Abs{x-\PEAK} - t$, for some values $\PEAK, t$.
Finally, the normalization $\int_0^1 \Fe{x} dx = 0$ implies that
$t = \half + \PEAK^2 - \PEAK$, completing the proof of Theorem
\ref{thm:worstcasefunction}.

\begin{figure}
\begin{center}
\epsfxsize=11.5cm
\epsffile{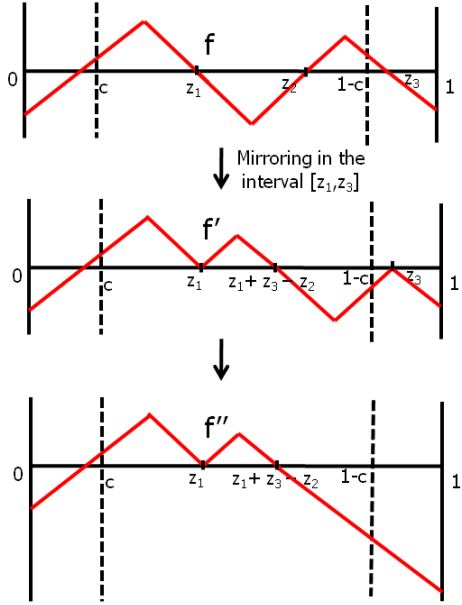}
\caption{Mirroring \FE in $[z_1,z_3]$ \label{fig:mirroring}}
\end{center}
\end{figure}

\section{Future Work}
Our work is a first step toward obtaining optimal (as opposed to
asymptotically optimal) randomized algorithms for choosing $k$ sample locations
to estimate an aggregate quantity of a function \FE. The most obvious
extension is to extend our results to the case of estimating
the average using $k$ samples. It would be interesting whether approximation
guarantees for the $k$-median problem (the deterministic counterpart)
can be exceeded using a randomized strategy.

Also, our precise characterization of the optimal sampling distribution
for functions on the $[0,1]$ interval should be extended to
higher-dimensional continuous metric spaces.
Another natural direction is to consider other aggregation goals,
such as predicting the function's maximum, minimum, or median.
For predicting the maximum from $k$ deterministic samples, a
2-approximation algorithm was given in \cite{das}, which is is best
possible unless P=NP. However, it is not clear if equally good
approximations can be achieved for the randomized case.
For the median, even the deterministic case is open.

On a technical note, it would be interesting whether finding the best
sampling distribution for the single sample case is NP-hard.
While we presented a PTAS
in this paper, no hardness result is currently known.

\subsubsection*{Acknowledgments}
We would like to thank David Eppstein, Bobby Kleinberg and Alex
Slivkins for helpful discussions and pointers, and anonymous referees
for useful feedback on previous versions.

\bibliographystyle{plain}
\bibliography{sampling}

\end{document}